\documentclass[nonacm,sigconf]{acmart}

\AtBeginDocument{%
  }

\usepackage{amsmath}
\usepackage{caption}
\usepackage{multirow}
\usepackage[linesnumbered,lined,boxed,commentsnumbered,ruled,longend, vlined]{algorithm2e}

\newcommand{\T}{{\mathcal T}}

\usepackage{booktabs} 
\usepackage{enumitem}
\usepackage{mathtools}
\usepackage{todonotes}
 \usepackage{nicematrix}

 \usepackage{mathtools, nccmath}

\newcommand{\ra}[1]{{\color{black}#1}}

\newcommand{\sss}[1]{{\color{black}#1}}
\newcommand{\sssp}[1]{{\color{black}#1}}

\begin{document}

\title{Near-Optimal Stability for Distributed Transaction Processing in Blockchain Sharding}



\author{Ramesh Adhikari}
\orcid{0000-0002-8200-9046}
\affiliation{%
  \institution{School of Computer \& Cyber Sciences\\Augusta University}
  \city{Augusta}
  \state{Georgia}
  \country{USA}
  \postcode{30912}
  }
\email{radhikari@augusta.edu}

\author{Costas Busch}
\orcid{0000-0002-4381-4333}
\affiliation{%
  \institution{School of Computer \& Cyber Sciences\\Augusta University}
  \city{Augusta}
  \state{Georgia}
  \country{USA}
  \postcode{30912}
}
\email{kbusch@augusta.edu}

\author{Dariusz R. Kowalski}
\orcid{0000-0002-1316-7788}
\affiliation{%
  \institution{School of Computer \& Cyber Sciences \\ Augusta University}
  \city{Augusta}
  \state{Georgia}
  \country{USA}
  \postcode{30912}
}
\email{dkowalski@augusta.edu}









\begin{abstract}
In blockchain sharding, 
$n$ processing nodes are divided into $s$ shards, and each shard processes transactions in parallel. A key challenge in such a system is to ensure system stability 
for any ``tractable'' pattern of generated transactions; 
this is modeled by an adversary generating transactions with a certain rate of at most $\rho$ and burstiness $b$.
This model captures worst-case scenarios and even some attacks on transactions' processing, e.g., DoS.
A stable system ensures bounded transaction queue sizes and bounded transaction latency.
It is known that the absolute upper bound 
on the maximum injection rate 
for which any scheduler could guarantee bounded queues and latency of transactions is $\max\left\{ \frac{2}{k+1},  \frac{2}{ \left\lfloor\sqrt{2s}\right\rfloor}\right\}$, 
where $k$ is the maximum number of shards that each transaction accesses. Here, we first provide a single leader scheduler that guarantees stability under injection rate $\rho \leq \max\left\{ \frac{1}{16k}, \frac{1}{16\lceil \sqrt{s} \rceil}\right\}$. 
Moreover, we also give a distributed scheduler with multiple leaders that guarantees stability under injection rate $\rho \leq \frac{1}{16c_1 \log D \log s}\max\left\{ \frac{1}{k}, \frac{1}{\lceil \sqrt{s} \rceil} \right\}$,
where $c_1$ is some positive constant and $D$ is the diameter of shard graph $G_s$.
This bound is within a poly-log factor from the optimal injection rate,
and significantly improves the best previous known result
for the distributed setting \cite{adhikari2024spaastable}.
\end{abstract}



\maketitle

\section{Introduction}
\label{sec:introduction}
Blockchain provides several unique features to transaction processing and storing, such as transparency, immutability, and fault tolerance~\cite{survey-of-onsensus}, and has been used in many applications~\cite{mcghin2019blockchain,azzi2019power}. However, despite these advantages, it suffers from scalability issues, low throughput, and high latency.
One of the methods to improve scalability is the sharding technique~\cite{Elastico,Rapidchain,Byshard}, 
which divides the entire blockchain network into smaller groups of nodes, called {\em shards}, and allows them 
to process transactions in parallel.

This work aims to improve
the stability of transaction processing by blockchains with shards.
Stability of a system is defined by maintaining bounded buffers (queues) and bounded latency of transactions' processing in an arbitrarily long run.
We focus on worst-case scenarios and thus consider an adversarial model of transactions' generation: new transactions could be generated (also called injected or arriving) at any time as long as certain injection rate $\rho$ and burstiness $b$ are preserved 
\cite{borodin2001adversarial,adhikari2024spaastable}.
This model also captures Denial of Service (DoS) attacks~\cite{raikwar2021attacks,nguyen2022denial}, in blockchain sharding systems.
The goal is to guarantee stability (i.e., bounded queues and latency) for as large injection rate $\rho$ as possible and for any fixed~burstiness~$b$.

Similarly to 
\cite{adhikari2024spaastable,adhikari2024fast,adhikari2023lockless}, we consider the blockchain sharding system with $n$ nodes that are divided into $s$ shards. Each shard has its own local blockchain and holds a subset of the accounts and states related to those accounts. We consider the partial-synchronous model~\sss{\cite{dwork1988consensus}}, which is more practical than the synchronous communication model considered in~\cite{adhikari2024spaastable,adhikari2024fast}. We assume that a transaction $T_i$ might be generated in any shard $S_i\in s$, which we call the {\em home shard} of transaction~$T_i$. Moreover, considering the sharding model~\cite{adhikari2023lockless,adhikari2024spaastable,adhikari2024fast}, each transaction $T_i$ could access accounts from at most $k$ shards, which we call the {\em destination shards}~for~$T_i$. Transaction $T_i$ can be split into subtransactions $T_{i,j}$,
where each subtransaction is sent to the respective {\em destination shard} that holds the corresponding account for the commitment. Similarly to 
\cite{adhikari2024spaastable,adhikari2024fast}, the maximum distance between the home shard 
and destination shards of a transaction 
is at most $d\leq D$ in the shard graph~$G_s$, where $D$ is the diameter of $G_s$. 

We assume that shards 
have queues that store pending transactions.
A shard 
picks transactions from its queue, splits them into subtransactions, and sends them to destination shards for processing in parallel. A conflict occurs when shards pick transactions that access the same account in destination shards.
Such conflicts prevent the transactions from being committed concurrently, and forces the conflicting transactions to be committed in a serialized order. Our proposed scheduler coordinates with the shards to determine a schedule with a sequence of conflict-free sets of transactions that commit concurrently, which ultimately imposes an efficient serialization of conflicting transactions. 


\renewcommand{\arraystretch}{1.4}
 \begin{table*}[t!]
\centering
\smaller[1]
\begin{tabular}{|l|l|l|}
\cline{1-3} 
    &{\qquad {\bf Previous Work~\cite{adhikari2024spaastable}}}  & \qquad \qquad{\bf Proposed Results} \\
\cline{2-3}
  & \qquad \qquad Synchronous & \qquad \qquad Partially-synchronous \\ 
  \cline{2-3}
 \multirow{2}{*}{{  }}  & 
    \multicolumn{2}{l|}{ 
    {\qquad \qquad \qquad \qquad \qquad   {\bf Single Leader Scheduler}}}
\\
 
  \cline{1-3}

   Txn injection rate $\rho$ &$\rho \leq \max\left\{ \frac{1}{18k}, \frac{1}{\lceil 18 \sqrt{s} \rceil} \right\}$ & $\rho \leq \max\left\{ \frac{1}{16k}, \frac{1}{16\lceil \sqrt{s} \rceil} \right\}$ \ (Th~\ref{theorem:single-leader}) \\ 
    \cline{1-3}
 Pending queues size 
   & $\le 4bs$ &  $\le 2(2b+\rho\cdot 48\mathfrak{D})s$ \ (Th~\ref{theorem:single-leader}) \\ 
   \cline{1-3}
    Latency 
    & $\le 36b \cdot \min \{k, \lceil \sqrt{s} \rceil \}$ &  $\le 32b\cdot \min \{k, \lceil \sqrt{s}\rceil\}+96\mathfrak{D}$ \ (Th~\ref{theorem:single-leader})  \\ 
   \cline{1-3}
  Shards network & Clique graph with unit distance & General graph with diameter $D$\\ 
\cline{1-3}

 \multirow{2}{*}{{  }}   & \multicolumn{2}{l|}{ {\qquad \qquad \qquad \qquad {\bf Multiple Leaders Scheduler}}}
    \\
  \cline{1-3}

   Txn injection rate $\rho$ & $\rho \leq \frac{1}{c'd \log^2 s} \cdot \max\left\{ \frac{1}{k}, \frac{1}{\sqrt{s}} \right\}$&  $\rho \leq \frac{1}{16 c_1 \log D \log s}\max\left\{ \frac{1}{k}, \frac{1}{\lceil \sqrt{s} \rceil} \right\}$ \ (Th~\ref{theorem:multi-leader-scheduler})\\ 
    \cline{1-3}
    Pending queues size 
   & $\le 4bs$& $\le 2(2b+ \rho \cdot 48 c_1 \mathfrak{D} \log D \log s)s$ \ (Th~\ref{theorem:multi-leader-scheduler})  \\ 
   \cline{1-3}
    Latency 
    & $\le 2 \cdot c'bd \log^2 s \cdot \min \{k, \lceil \sqrt{s}\rceil\}$ &  $32c_1b \log D \log s \cdot \min\{k, \sqrt{s}\} $ \ \ 
    \\
    &&$+ 96 c_1 \mathfrak{D} \log D \log s$ \ (Th~\ref{theorem:multi-leader-scheduler})
    \\ 
   \cline{1-3}
  Shards network &\multicolumn{2}{l|}{ {\qquad \qquad \qquad \qquad General graph with diameter $D$}} 
  \\ 
\cline{1-3}

\end{tabular}
\caption{
Comparison of our schedulers with the best-known results in~\cite{adhikari2024spaastable}, where $d \leq D$ denotes
the maximum distance of any transaction from its home shard to the shards it will access, $\mathfrak{D}$ denotes an upper bound on the local processing time and the communication delay between any two shards in our proposed protocol \sssp{(is at least $d$)}, and  $c', c_1$ are some constants implied by analysis.}
\label{tbl:contribution-summary}
\end{table*}

\paragraph*{\bf Contributions.}
We improve the best-known transaction injection rates~\cite{adhikari2024spaastable}, which guarantee stability, and we achieve this by employing a specific event-driven approach in designing the schedulers. 
Another advantage of our approach is that it does not require full synchrony.
Moreover, the previous work~\cite{adhikari2024spaastable} assumes a clique graph with unit distances for a single leader scheduler, whereas we consider a general graph with a diameter $D$. 
For such graphs, the scheduler in~\cite{adhikari2024spaastable} would stabilize for $D$ times smaller injection rate.
Similarly, in a multi-leader setting, 
our scheduler allows stability for $d$ times higher injection rate of transactions than in~\cite{adhikari2024spaastable}.
The summary of our contributions and their comparison to~\cite{adhikari2024spaastable} is shown in Table~\ref{tbl:contribution-summary}, 
which we also describe briefly as follows:
\begin{itemize}
    \item First, we provide a single leader scheduler in partially-synchronous setting, which guarantees stability under transaction injection rate $\rho \leq \max\{ \frac{1}{16k}, \frac{1}{16\lceil \sqrt{s} \rceil} \}$.
    The queue size of this scheduler's pending transactions 
    is bounded by $2(2b+\rho\cdot 48\mathfrak{D})s$ and the latency is bounded by $32b\cdot \min \{k, \lceil \sqrt{s}\rceil\}+96\mathfrak{D}$, 
    where $\mathfrak{D}$ denotes an upper bound on the combined 
    delay of local processing (including consensus in a shard) and communication between~any~shards.
    \item Next, we provide a multi-leader distributed scheduler, where multiple non-overlapping shard leaders process the transactions in parallel. This scheduler guarantees  stability for transaction injection rate $\rho \leq \frac{1}{16c_1 \log D \log s}\max\left\{ \frac{1}{k}, \frac{1}{\lceil \sqrt{s} \rceil} \right\}$, for which queue sizes are bounded by $2(2b+\rho\cdot 48c_1\mathfrak{D} \log D \log s)s$ and transaction latency is bounded by $32c_1b \log D \log s \cdot \min\{k, \sqrt{s}\} + 96 c_1 \mathfrak{D} \log D \log s$,  where $c_1$ is a positive constant coming from the analysis.
\end{itemize}

Note that in \cite{adhikari2024spaastable}, for multiple leaders, the injection rate $\rho$  is suboptimal because it inversely depends (linearly) on the parameter $d$. In our result, Theorem \ref{theorem:multi-leader-scheduler}, the injection rate $\rho$ does not depend linearly on $d$, but only on $\log D$; hence, our bound on injection rate $\rho$ is near-optimal (within poly-log factors).

\paragraph*{\bf Paper Organization.} 
Section \ref{sec:related-work} provides the related work, and Section \ref{preliminaries} describes the preliminaries for this study and the sharding model.
Section~\ref{sec:single-leader} presents a single leader scheduler.
Section~\ref{sec:multiple-leaders} generalizes the techniques to a multi-leader scheduler in the distributed setting. We provide the conclusion in Section~\ref{sec:conclusion}.


\section{Related Work}
\label{sec:related-work}
Several blockchain sharding protocols~\cite{Elastico,Rapidchain,Byshard,Zilliqa} have been proposed to solve the scalability issue of the blockchain. Although these protocols provide better throughput and process the transactions with minimum latency, none of these protocols provide stability analysis under adversarial transaction generation. Similarly, there is recent work on transaction scheduling in blockchain sharding~\cite{adhikari2024fast}, but this paper does not provide stability analysis; they only focus on the performance-oriented scheduling problem. 

The only other known previous work that provides a stability analysis for blockchain sharding appears in~\cite{adhikari2024spaastable}. However, the communication model considered in their protocol is synchronous, which may not accurately represent real distributed blockchain scenarios. Moreover, their transaction injection rate is not close to an optimal one for the distributed setting (i.e. it deviates from the optimal injection rate by a factor of $1/d$). In this work, we consider a partially-synchronous model, and we also provide near-optimal transaction injection rate bounds (that do not depend on $d$) for which the system remains stable. Hence, we significantly improve on the previous work.
Moreover, in~\cite{adhikari2024spaastable}, the single leader model was analyzed under a synchronous setting with a clique graph of unit distance. However, such assumptions may not be practical in real-world blockchain networks. Therefore, we extend the model here not only by allowing partially synchronous communication but also by considering a general graph structure with a diameter $D$ and maximum communication delay $\mathfrak{D}$, where shards can be connected to each other in a more flexible topology.

The adversarial model we consider here was first introduced in the context of adversarial queuing theory by Borodin 
{\em et al.}~\cite{borodin2001adversarial} for stability analysis in routing algorithms, where packets are continuously generated over time. Adversarial queuing theory~\cite{borodin2001adversarial} provides a framework for establishing worst-case injection rate bounds, and it was recently used to analyze blockchain sharding environments with unpredictable transaction injections~\cite{adhikari2024spaastable}. Adversarial queuing theory has also been applied to various dynamic tasks in communication networks and channels~\cite{bender2005adversarial,chlebus2009maximum,chlebus2012adversarial}. 
Moreover, \sss{in~\cite{busch2023stable}, the authors introduced a stable scheduling algorithm for software transactional memory systems under adversarial transaction generation. Their model assumes a synchronous communication model, it allows objects to move across nodes, and transactions are executed once the required objects become available. However, in the blockchain sharding model, objects are static in shards, and transaction commitment requires confirmation from all the respective involved shards. Therefore, the results in~\cite{busch2023stable} do not directly apply to our transaction scheduling model.}

\section{Technical Preliminaries} 
\label{preliminaries}
Similar to the previous works~\cite{adhikari2024spaastable,adhikari2023lockless},
we consider a blockchain system with 
$n$ participating nodes, which are divided into 
$s$ shards $S_1, S_2, \cdots S_s$,
where each shard $S_i$ is a subset of nodes, i.e., 
$S_i \subseteq \{1, \ldots, n\}$.
The shards are pairwise disjoint, forming a partition of the $n$ nodes,
that is, for any $i \neq j$, $S_i \cap S_j = \emptyset$, and $n = \sum_i |S_i|$.
Similarly to~\cite{adhikari2024spaastable}, we assume that shards communicate via message passing, where all non-faulty nodes within a shard must reach consensus on each message before sending it to another shard. To achieve this, each shard executes a consensus algorithm before transmitting messages (e.g., using PBFT \cite{PBFT} in the shard). 
We denote $n_i = |S_i|$ is the total number of nodes in shard $S_i$ (the size of $S_i$), and by $f_i$ is
the number of faulty (Byzantine) nodes in $S_i$. 
For Byzantine fault tolerance and security guarantees, we assume that each shard $S_i$ satisfies the condition $n_i > 3 f_i$.
Moreover, similarly to ~\cite{adhikari2024spaastable,Byshard,adhikari2023lockless,adhikari2024fast}, we assume that the system consists of a set of shared accounts $\mathcal{O}$ (we also refer to them as objects). 
The set 
$\mathcal{O}$ is partitioned into disjoint subsets 
$\mathcal{O}_1, \ldots, \mathcal{O}_s$ where $\mathcal{O}_i$
represents the collection of objects handled by shard $S_i$.
Each shard $S_i$ maintains its own local blockchain (ledger) based on the subtransactions it processes for $\mathcal{O}_i$.



\paragraph{\bf Transactions and Subtransactions.}
For a transaction $T_i$ the respective {\em home shard} $S_i$ 
is where it gets injected into.
Similar to~\cite{adhikari2024spaastable,adhikari2023lockless,adhikari2024fast,Byshard}, 
a transaction $T_i$ is represented as a collection of subtransactions $T_{i,a_1},\ldots,T_{i,a_j}$
where subtransaction $T_{i,a_l}$ accesses only objects in $\mathcal{O}_{a_l}$.
Thus, subtransaction $T_{i,a_l}$ has a respective {\em destination shard} $S_{a_l}$
where it will ultimately be appended in the local ledger. 
%
%
An intra-shard transaction is a special case 
where the transaction accesses only accounts in the home shard. 
If a transaction accesses accounts outside the home-shard, then we refer to it as a {\em cross-shard} transaction.
In~\cite{Rapidchain}, it was observed that cross-shard transactions are most frequent ($99.98\%$ of all transactions for $16$ shards), thus we focus on the general case of cross-shard transactions.

Two transactions $T_i$ and $T_j$ {\em conflict} if they access a common object 
$o_k \in \mathcal{O}$ concurrently and at least one of them modifies (updates) the value  of $o_k$. For safe transaction processing, the corresponding subtransactions of $T_i$ and $T_j$ must be serialized in the same order across all involved (destination) shards.

\begin{example}
Similar to the works~\cite{Byshard,adhikari2024spaastable,adhikari2023lockless}, let us consider a transaction $T_1$ that accesses multiple accounts distributed across different shards. Suppose $T_1$ represents the following banking transaction: {\em ``Transfer 100 from Apsar's account to Prakriti's account, only if Jagan's account has at least 200.''}. Let us assume that the accounts of Apsar, Prakriti, and Jagan reside on different shards $S_a$, $S_p$, and $S_j$, respectively.

In our protocol, each home shard sends its transaction to the corresponding leader shard. Then the leader shard of $T_1$ splits the transaction $T_1$ it into the following subtransactions, with respect to the destination shards it accesses:
\end{example}

\begin{itemize}[itemindent=0.7cm]
    \item[$T_{1,a}$] {Condition:} Check if Apsar's account has at least 100.\\
    {Action:} Deduct 100 from Apsar's account.
    \item[$T_{1,p}$] {Action:} Add 100 to Prakriti's account.
    \item[$T_{1,j}$] {Condition:} Check if Jagan's account has at least 200.
\end{itemize}

To execute $T_1$, the leader shard first requests the current state of all involved accounts from the corresponding destination shards ($S_a$, $S_p$, and $S_j$) in parallel. Once the leader receives all the account states, it checks the validity of the transaction locally. If all conditions are satisfied (i.e., Jagan has at least 200 and Apsar has at least 100, and the involved accounts exist), then the leader \emph{pre-commits} $T_1$; otherwise, it aborts $T_1$.

Next, the leader constructs pre-committed, ordered batches of subtransactions for each destination shard, possibly combining $T_1$ with other transactions targeting the same shard (because the leader shard might reserve multiple transactions from home shards for scheduling). These batches are sent in parallel to the corresponding destination shards. Upon receiving their batch, each destination shard performs consensus on the batch and appends it to the local ledger.
In this way, the transaction $T_1$ is either consistently committed across all relevant shards or aborted entirely and preserves atomicity and consistency.

\paragraph{\bf Communication Model, Emit, and Events.} 
We  \sss{consider a partially synchronous communication model~\cite{dwork1988consensus}, which assumes that there exists a Global Stabilization Time (GST) after which all messages are guaranteed to be delivered within a known bounded delay ($\mathfrak{D}$), and this delay is known to all shards.
\sssp{
Similar to previous works in~\cite{Byshard,adhikari2024spaastable,adhikari2025efficiency}, we adopt cluster sending protocols from  \textit{the fault-tolerant cluster-sending problem}~\cite{clustersending} and Byshard~\cite{Byshard}, where shards run consensus before sending a message and assume that sending a message between two shards incurs a delay of at most~$\mathfrak{D}$ (after GST).
    Namely, for communication between {\em neighboring shards} $S_1$ and $S_2$, a set $A_1 \subseteq S_1$ of $2f_1+1$ nodes in $S_1$ and a set $A_2 \subseteq S_2$ of  $f_2 + 1$ nodes in $S_2$ are chosen (where $f_i$ is the maximum number of faulty nodes in shard $S_i$). Each node in $A_1$ is instructed to broadcast the message to all nodes in $A_2$. This ensures that at least one non-faulty node in $S_1$ will send the correct message value to a non-faulty node in $S_2$.  This mechanism tolerates up to $f_i < |S_i|/3$ Byzantine faults while guaranteeing reliable cross-shard communication. 
\newline
For communication between {\em non-neighboring shards} $S_0$ and $S_3$,
routing between them is fault-tolerant because of the following recursive argument. Consider a path between $S_0$ and $S_3$ is the shard graph $G_s$. Each edge $(S_1,S_2)$ on this path transfers messages reliably, based on the above argument about communication between neighboring shards $S_1,S_2$. Hence, by simple induction on the subsequent edges of this path, the destination shard receives the same message as the shard $S_0$ agreed to send to $S_3$ when executing consensus (inside its cluster $S_0$).
Please note that consensus is only needed, and applied, inside shards (each time independently on executions within other shards), not globally in the whole shards' network $G_s$. Since each shard is a clique (hence, no other assumptions for running consensus are needed apart of limiting the number of faulty nodes by third part of the shard), we do not need any additional connectivity assumptions on the whole~graph~$G_s$
}
Transactions that are generated before GST are accounted in the burstiness parameter $b$.}
The parameter $\mathfrak{D}$ is determined by two delays:
(i) the local consensus worst time ($\delta_{cons}$) within a shard,
which we normalize to be at most one time unit (i.e., $\delta_{cons} = 1$); 
and (ii) the worst network delay ($\delta_{comm}$), network delay is the time it takes for a message (e.g., transactions or state request) to travel from one shard to another in the shard graph $G_s$, which can be arbitrary depending on the network delays.
Hence, $\mathfrak{D} = \delta_{cons} + \delta_{comm}$.

We define two key terms 
describing
how messages and actions propagate within our scheduling algorithm: (i) {\em Emit}: Represents the act of sending a message or request asynchronously. 
E.g.,
if shard $S_i$ receives a new transaction $T_k$ and there is some shard, say $S_j$, 
destined to process the transaction $T_k$,
then shard $S_i$ emits a send request $T_k$ to shard $S_j$. (ii) {\em Event}: Represents the reception of an emitted message. This triggers an appropriate response. Events occur when messages reach their destination after a bounded delay (at most $\mathfrak{D}$) since their emit. In our example, upon receipt of the message, the destination shard $S_j$ performs 
specified
action to put the transaction into~the~processing~schedule.

\paragraph{\bf Adversarial Model of Transactions' Generation.}
\label{subsec:adversarial-model}
The adversarial model~\cite{borodin2001adversarial,adhikari2024spaastable} we consider here models the process in which users  generate and inject 
transactions into the sharded blockchain system continuously and arbitrarily with transaction injection rate at most $\rho$ and burstiness $b$, where $0 < \rho \leq 1$ and $b>0$. We 
assume
that each injected transaction accesses at most $k$ shards, and 
adds a unit congestion to each accessing destination shard. \ra{The injected transactions are stored in the pending queue of the shard, and we examine and analyze the combined pending queue size throughout the system.} The adversary model is constrained in such a way that over any contiguous time interval of length 
$t > 0$, the congestion at any shard (number of transactions that access accounts in the shard) does not exceed 
$\rho t + b$ transactions. 
Here, $\rho$ represents the maximum transaction injection rate per unit time at each destination shard, and 
$b$ represents the adversary's ability to generate a burst (maximum number) of transactions within any unit of time,  \sss{including transactions generated before the Global Stabilization Time (GST) due to partial synchrony}.

\sssp{In the literature, there are basically two adversarial models: one is called a window type adversary (see e.g., \cite{borodin2001adversarial}), and the one considered in our work is called a leaky bucket adversary, see e.g.,~\cite{Andrews2001}. As discussed in~\cite{chlebus2012adversarial}, the former is a logically restricted form of the latter; however, in most considered applications of this adversarial framework, they are equivalent (see e.g., \cite{ROSEN2002237}).
An additional reason why we chose the leaky bucket model over the window adversarial model is the ability to easily account for transactions injected before the Global Stabilization Time (GST) as part of the burst (which needs to be resolved in the considered period after the GST). One of the advantages of having $b$ in our performance formulas is, among others, that we can see how the transactions injected before the GST may influence system performance even after the GST.}

\sssp{Similar to the literature in this topic~\cite{adhikari2024spaastable}, our current analysis focuses on the congestion of transaction processing capacity within shards. Note that transactions are processed only in the home and destination shards (or in the leader shard), but not in transit. 
        The delays based on the communication congestion, caused by transiting transactions between the home and destination shards, are accounted for by the parameter $\mathfrak{D}$ (see above, Communication Model). The parameter $\mathfrak{D}$ upper bounds both processing and network delays. Notation-wise, the delay caused by communication congestion is included in  $\delta_{\mathrm{comm}}$, which is the part of the parameter $\mathfrak{D}$.
        There was a vast number of papers, started by~\cite{borodin2001adversarial,Andrews2001}, analyzing the stability and delays caused by communication congestion. 
        For instance,~\cite{Alvarez2005} applied adversarial models to capture phenomena related to the routing of packets with varying priorities and failures in networks.
Authors in~\cite{1316091}~addressed the impact of link failures on the stability of communication algorithms by way of modeling them in
adversarial terms.
\newline
        Depending on the communication model and setting (e.g., failures), our parameter $\delta_{\mathrm{comm}}$ could be replaced by the bounds obtained in those papers.
In this view, using a generic parameter $\delta_{\mathrm{comm}}$ actually makes our paper even more general, because it is not dependent on specific communication model -- as long as there exists a bound on $\delta_{\mathrm{comm}}$ in that model, one could substitute $\delta_{\mathrm{comm}}$ by that bound in our analysis.}

\section{Event-Driven Scheduler with Single Leader}
\label{sec:single-leader}
In this section, we design and analyze an event-driven single leader transaction scheduler that operates in a partial-synchronous communication model.

\noindent
{\bf High-level Idea:} In our approach, 
after a transaction $T_i$ is generated at the home shard, 
then, the home shard forwards the transaction to the leader shard $S_{ldr}$, which is responsible for scheduling the transactions. Upon receiving transaction information from multiple home shards, the leader shard determines the schedule length, i.e., the time required to process and commit these transactions. $S_{ldr}$ constructs a transaction conflict graph $G_\T$ and uses a greedy vertex coloring algorithm to determine schedule length and conflict-free commit order. If the computed schedule length is at least as long as the 
previously computed
schedule length, the leader initiates the scheduling process. In order to properly determine the schedule length, the leader shard first requests the latest account states from the destination shards that hold the relevant accounts. Once the leader shard receives the account state information, 
the leader pre-commits the transactions according to the color they get in $G_\T$ and creates a batch order for each destination shard; this is possible because the leader has gathered all required account state information from the destination shards. This pre-committed batch is then sent to the respective destination shards. After receiving the pre-committed batch of subtransactions, each destination shard runs locally a consensus algorithm to reach an agreement on the received order and 
to append the respective transactions into their local blockchain. 

\begin{algorithm*}[ht!]
\smaller[1]
\caption{{\sc EdSlScheduler}}
\label{alg:centralized-optimized-scheduler}

$S_{ldr}$: Leader shard;
$PQ_{ldr}$: Pending transactions queue in leader shard\;
$SQ_{ldr}$: Scheduled transactions queue in leader shard\;
$\mathcal{V}$: Version number of scheduling event, initially $\mathcal{V}=0$, and it increments\;
$LastEventSchLength$: Last scheduled event length, initially, 0\;

$BatchTxn(\mathcal{V})$: Batch transaction of version $\mathcal{V}$\;
$PrecommitSubTxnBatch(\mathcal{V}, S_j)$: Pre-committed subtransactions batch for shard $S_j$\;

\BlankLine
\textcolor{black}{\tcc{\bf Event on Home Shards:}}
\SetKwBlock{EventOne}{\bf Upon event: {\em new transaction $T_i$ generated on any home shard $S_i$}}{}
\EventOne{
    Shard $S_i$ {\em \textbf{Emits} transaction\_send($T_i$)} for $T_i$ to be received by leader shard $S_{ldr}$\;
}

\BlankLine
\textcolor{black}{\tcc{\bf Events on Leader Shard $S_{ldr}$:}}
\SetKwBlock{EventTwo}{\bf Upon event: {\em transaction\_send($T_i$)}}{}
\EventTwo{
    $S_{ldr}$ appends $T_i$ to pending queue $PQ_{ldr}$\;

    $S_{ldr}$ constructs transaction conflict graph $G_{\T}$ and colors it using greedy vertex coloring to determine schedule length $\lambda$ for transactions in $PQ_{ldr}$\;

    \If{$\lambda \geq LastEventSchLength$}{
       \tcp{\textcolor{black}{Scheduling Event Triggered; $S_{ldr}$ does the following:}}
      Update the version number: $\mathcal{V} \leftarrow \mathcal{V}+1$\;

      Move transactions  from $PQ_{ldr}$ to $SQ_{ldr}$ with version $\mathcal{V}$\;
      
      Create batch transaction $BatchTxn(\mathcal{V})$\;
      Assign the version $\mathcal{V}$ to the transaction conflict graph $G_\T$  as $G_{\T^{\mathcal{V}}}$\;
      $LastEventSchLength \leftarrow \lambda$\;

      Determine destination shards for each $T_i \in BatchTxn(\mathcal{V})$\;
      
      {\em \textbf{Emit} account\_state\_request($BatchTxn(\mathcal{V})$)} to be received by respective destination shards\;
    }
}

\SetKwBlock{EventFour}{\bf Upon event: {\em account\_state\_response($BatchTxn(\mathcal{V})$)}}{}
\EventFour{
    \If{All account states received for $BatchTxn(\mathcal{V})$}{
        According to the commit order, which was set previously in the transaction conflict graph $G_{\T^{\mathcal{V}}}$ of $BatchTxn(\mathcal{V})$, $S_{ldr}$ does the following\;

        \ForEach{color group $C_k$ in $G_{\T^{\mathcal{V}}}$}{
            Pre-commit or abort transactions $T_i\in C_k$\ by checking transactions and account state conditions according to the commit order determined by coloring\;
            If $T_i$ is pre-committed, split $T_i$ into subtransactions and create pre-committed subtransactions batch $PrecommitSubTxnBatch(\mathcal{V}, S_j)$ for each destination shard $S_j$\;
        }

        {\em \textbf{Emit} send\_precommit\_batch($PrecommitSubTxnBatch(\mathcal{V}, S_j)$)} to the corresponding destination shards $S_j$, for each shard $S_j$\;
    }
}

\SetKwBlock{EventEleven}{\bf Upon event: {\em final\_commit\_response($T_{i,j},\mathcal{V}$)}}{}
\EventEleven{
   If for all subtransactions $T_{i,j}$ of $T_i$, {\em final\_commit\_response($T_{i,j},\mathcal{V}$)} was emitted, remove $T_i$ from $SQ_{ldr}$\;
     If $T_i \in LastEventSchLength$ then $LastEventSchLength \leftarrow LastEventSchLength - |T_i|$\;
}

\BlankLine
\textcolor{black}{\tcc{\bf Events on Destination Shards:}}
\SetKwBlock{EventNine}{\bf Upon event: {\em account\_state\_request($BatchTxn(\mathcal{V})$)}}{}
\EventNine{
    \If{Requested account is available in shard $S_j$}{
        {\em \textbf{Emit} account\_state\_response($BatchTxn(\mathcal{V})$)} to $S_{ldr}$ and set respective account to {\em occupied}\;
    }
    \Else{Wait until the account is available and Emit account state response after wait\;}
}

\SetKwBlock{EventSix}{\bf Upon event: {\em send\_precommit\_batch($PrecommitSubTxnBatch(\mathcal{V}, S_j)$)}}{}
\EventSix{
    $S_j$ reaches consensus on $PrecommitSubTxnBatch(\mathcal{V}, S_j)$ and appends it to the local blockchain;
    and sets respective account to {\em available}\;
    $S_j$ {\em \textbf{Emit} final\_commit\_response($T_{i,j},\mathcal{V}$)} to $S_{ldr}$ shard\;
}

\end{algorithm*}

\noindent
{\bf Algorithm Details: }
Our Event-Driven Single-Leader Scheduler ({\sc EdSlScheduler}) Algorithm~\ref{alg:centralized-optimized-scheduler} is structured into three key phases: (i) transaction initiation at home shards, (ii) scheduling at the leader shard, and (iii) finalization at destination shards. Let us describe each phase in detail.

\noindent
{\bf \em (i) At the home shards:}  When a new transaction $T_i$ is generated,~it~is~immediately sent to the designated leader shard $S_{ldr}$ by emitting an event {\em transaction\_send($T_i$)}. We assume that all shards are aware of the leader’s ID.

\noindent
{\bf \em (ii) At the leader shard:}  At the leader shard $S_{ldr}$ the transactions are maintained in two queues. The first one is a pending transaction queue ($PQ_{ldr}$) that stores received but unscheduled transactions. The second one is a scheduled transaction queue ($SQ_{ldr}$) that tracks transactions that are scheduled but awaiting for final commitment. Upon receiving a new transaction $T_i$ due to the emitting an event {\em transaction\_send($T_i$)} from a home shard, the leader appends transaction $T_i$ to $PQ_{ldr}$. As the pending queue may contain transactions from multiple shards, the leader calculates the schedule length $\lambda$, which represents an upper bound on the total processing time required for all transactions in $PQ_{ldr}$. To determine $\lambda$ and conflict-free commit order, leader shard $S_{ldr}$ constructs a transaction conflict graph $G_\T$ from transactions $PQ_{ldr}$ using greedy vertex coloring algorithm. If $\lambda$ is greater than or equal to the last recorded schedule length $LastEventSchLength$, the leader triggers a scheduling event. The leader then increments the version number  $\mathcal{V}$ of this scheduling event (initially set to $0$), moves transactions from $PQ_{ldr}$ to $SQ_{ldr}$, forms a transaction batch $BatchTxn(\mathcal{V})$, and set the version for $G_\T$ to $G_{\T^{\mathcal{V}}}$. It updates the last event schedule length $LastEventSchLength$ to $\lambda$ and determines the destination shards for each transaction. Next, the leader requests the latest account states from the destination shards by emitting an event {\em account\_state\_request($BatchTxn(\mathcal{V})$)}.

The versioning mechanism is introduced to handle the partial-synchronous and event-driven nature of our scheduling algorithm. Each scheduling event is associated with a version number $\mathcal{V}$ which is tagged in transactions, batch requests, and final committed transactions. This versioning is crucial for tracking the sequence of events and ensuring correct order processing. Since scheduling events can be triggered at any point in time, the leader shard may request account states from destination shards at different moments. Therefore, versioning helps the scheduler determine the correct state for each transaction batch. For instance, when a scheduling threshold (with respect to $\lambda$) is reached, the leader requests account states from destination shards. Without versioning, it would be ambiguous which batch the response corresponds to, potentially leading to inconsistencies. Additionally, versioning prevents conflicts by ensuring that account state modifications are correctly sequenced across multiple scheduling events.

Upon receiving the account state responses from the destination shards due to emit {\em account\_state\_response($BatchTxn(\mathcal{V})$} from destination shards, the leader shard pre-commit the transactions according to the execution order they get in colored conflict transaction graph $G_{\T^{\mathcal{V}}}$.
Transactions within the same color group $C_k$ are checked for pre-commit feasibility simultaneously based on account balance constraints (because same-color transactions are non-conflicting transactions, as they access different accounts and can be committed at the same time). This pre-commit in a leader is possible because the leader now has the latest account state information. Pre-committed transactions are divided into subtransactions according to the shards they access. Then, the leader creates pre-committed subtransaction batches $PrecommitSubTxnBatch(\mathcal{V}, S_j)$ for each destination shard $S_j$. Once all transactions in $G_{\T^{\mathcal{V}}}$ are processed, the leader sends the pre-committed batch orders to the respective destination shards 
by emitting
the event 
{\em send\_precommit\_batch($PrecommitSubTxnBatch(\mathcal{V}, S_j)$)}.

\noindent
{\bf \em (iii) At the destination shards:} The event {\em account\_state\_request($BatchTxn(\mathcal{V})$)} is triggered upon receiving an account state request from the leader shard. If the requested account is available, this indicates that there is no other transaction in progress (the account state is read but has not yet finally committed). Then the destination shard set the account to {\em occupied} to prevent conflicting access and responds with the current account state by emitting the event {\em account\_state\_response($BatchTxn(\mathcal{V})$)}. 
Otherwise, if an ongoing transaction is modifying the account, the destination shard delays its response until the current state is finalized.

Similarly, when a destination shard receives a pre-committed batch from the leader, the event {\em send\_precommit\_batch($PrecommitSubTxnBatch(\mathcal{V}, S_j)$)} is triggered. The destination shard then executes a consensus on the received batch and appends the finalized transactions to its local blockchain. After finalizing the transactions (subtransactions), the shard resets the account to available, allowing new transactions to access the account.
The destination shard then emits a final commit response {\em final\_commit\_response($T_{i,j},\mathcal{V}$)} to the leader, confirming successful commitment. Once the leader receives the final commit responses for all subtransactions of a given transaction $T_i$, 
it
removes $T_i$ from $SQ_{ldr}$. If $T_i$ belonged to the most recent scheduled batch, the leader reduces the last recorded schedule length by the number of committed transactions.

\subsection{Analysis for Algorithm \ref{alg:centralized-optimized-scheduler}}
Recall that each transaction accesses at most $k \geq 1$ out of $s \geq 1$ shards, and the burstiness is $b\geq 1$. Suppose $\mathfrak{D}$ is the upper bound on the local processing time and the communication delay between any two shards as described in Section~\ref{preliminaries}.
For the sake of analysis, we divide time into a sequence of intervals $I_1, I_2, \ldots,$ each of duration $\tau$ as defined below.
To simplify the analysis, we also introduce an artificial interval $I_0$ preceding the start of time, during which no transactions are present. Let the maximum transaction injection rate be $\rho'$, as defined below,
and let $\zeta$ 
denote the maximum number of transactions that can be generated in an interval:
\[
    \tau 
    :=
    16b\cdot \min\{k, \sqrt{s}\} + 48\mathfrak{D} \ , \ \ 
    \rho'
    :=
    \max\left\{ \frac{1}{16k}, \frac{1}{16\lceil \sqrt{s} \rceil} \right\} \ ,
\]

\[
    \zeta
    :=
    (2b+\rho\cdot 48\mathfrak{D})s \ .
\]


\begin{lemma}
\label{lemma:max-injection-rate}
In Algorithm \ref{alg:centralized-optimized-scheduler},
for transaction generation rate
$\rho \leq \rho'$
and $b\geq 1$,
in any interval $I_z$ of length $\tau$, where $z \geq 1$, there are at most $\zeta$ new transactions generated within  $I_z$.
\end{lemma}
\begin{proof}
    According to the definition of the adversary,
    during any interval $I_z$ of length $\tau$, the maximum congestion added to any shard is:
    \begin{equation}
    \label{eqn:transaction-generation-rate-alg1}
    \rho \cdot \tau + b \leq b + \rho \cdot 48\mathfrak{D} +b = 2b+ \rho\cdot 48\mathfrak{D}
    \ .
    \end{equation}

        Moreover, since each transaction accesses at least one shard and there are at most $s$ shards, 
        the total number of new transactions generated in  $I_z$~is~at~most
$     (2b+\rho\cdot 48\mathfrak{D})s = \zeta$.
\end{proof}

\begin{lemma}
\label{lemma:time-required-for-scheduler-when-event-triggered}
In Algorithm~\ref{alg:centralized-optimized-scheduler},
for transaction generation rate $\rho \leq \rho'$
and $b\geq 1$
in any interval $I_z$ of length $\tau$ (where $z \geq 1$), if there are at most $\zeta$ pending transactions, then it takes at most $\frac{\tau}{4}$ time units to commit all $\zeta$ transactions.
\end{lemma}

\begin{proof}
  We calculate the time required to process $\zeta$ transactions 
             based on the magnitude~of~$k$~and~$\lceil\sqrt{s}\rceil$.
             \begin{itemize}[leftmargin=*]
                 \item {\bf Case 1:  $k \leq \lceil \sqrt{s} \rceil$.}
    From Equation \ref{eqn:transaction-generation-rate-alg1} in the proof of Lemma~\ref{lemma:max-injection-rate}, the congestion of transactions per shard is $2b+\rho\cdot 48\mathfrak{D}$. 
    Since each transaction accesses at most $k$ shards, each 
    pending
    transaction conflicts with at most $((2b+\rho\cdot 48\mathfrak{D})-1)k$ other 
    pending
    transactions. Consequently, the highest degree $\Delta$ in the transactions conflict graph $G_\T$ 
    is $\Delta \leq ((2b+\rho\cdot 48\mathfrak{D})-1)k$. Hence, a greedy vertex coloring algorithm on $G_\T$ assigns at most $\Delta+1 \leq ((2b+\rho\cdot 48\mathfrak{D}) -1)k+1$ colors. Since $k \geq 1$, the total length of the required time interval to schedule and commit/abort the transactions is calculated as follows. Scheduling Algorithm \ref{alg:centralized-optimized-scheduler} consists of different steps, the initial step being to send the information about the transaction from the home shard to the leader shard, which takes at most $\mathfrak{D}$ time units.
    Similarly, the leader shard takes $1$ time unit for graph construction and coloring (as we consider local consensus time $\delta_{cons}=1$).
    Moreover, the leader asks the destination shards for the states of their accounts, 
    and the destination shards respond with the states of their accounts. This combined takes $2\mathfrak{D}$ time units. Additionally, leader shards take $\Delta+1$ time units to pre-commit each conflicting transaction serially, and $\mathfrak{D}$ time units to send the committed subtransaction batch to the destination shard. Finally, the destination shard takes $1$ time unit to reach consensus on the received order of subtransactions batch and append them on their local blockchain. Since $\rho \leq \rho'$, we have $\rho \cdot48\mathfrak{D} \leq \frac{3\mathfrak{D}}{k}$.
    Hence, the total combined required time length can be calculated as 
 %
 \begin{eqnarray*}
    \mathfrak{D} \ + &2\mathfrak{D}&+ 
    \ 1 + (\Delta+1) + \mathfrak{D} +1 \leq \\&4\mathfrak{D}& + \ 2 + 2bk +\frac{3\mathfrak{D}}{k}k -k +1
    \leq  2bk + 9\mathfrak{D} \leq \frac{\tau}{4}\ .
\end{eqnarray*}
\item   {\bf Case 2: $k> \lceil \sqrt{s}\rceil$.}
     We classify the $\zeta$ 
     pending transactions into two groups, the ``heavy'' transactions, which access more than $\lceil \sqrt{s}\rceil$ shards, and ``light'' transactions, which access at most $\lceil \sqrt{s}\rceil$ shards.
       The maximum number of heavy 
       transactions can be $(2b+\rho\cdot 48\mathfrak{D})\lceil \sqrt{s}\rceil$. If there were more, the total congestion 
       would be strictly greater than $(2b+\rho\cdot 48\mathfrak{D})\lceil \sqrt{s}\rceil \cdot \sqrt{s} \geq (2b+\rho\cdot 48\mathfrak{D})s$, which is not possible 
       \ra{because of Lemma~\ref{lemma:max-injection-rate}}.
    A coloring algorithm for the conflict graph $G_\T$ can assign each of the heavy transactions a unique color, which can be done using at most $\Gamma_1 =(2b+\rho\cdot 48\mathfrak{D})\lceil \sqrt{s}\rceil$~colors.
      
    The remaining transactions are light. Let $G'_\T$ be the subgraph of the conflict graph $G_\T$ 
   induced by
    the light transactions. Each light transaction conflicts with at most $((2b+\rho\cdot 48\mathfrak{D})-1)\lceil \sqrt{s}\rceil$
     other light transactions.
     Hence, 
     the degree of $G'_\T$ is at most $(2b + \rho \cdot 48 \mathfrak{D}-1)\lceil \sqrt{s}\rceil$.
     Thus, $G'_\T$ can be colored with at most $\Gamma_2 = (2b + \rho \cdot 48 \mathfrak{D}-1)\lceil \sqrt{s}\rceil + 1$~colors.

    Consequently, the conflict graph $G_\T$ can be colored with at most $\Gamma = \Gamma_1+\Gamma_2 = (2b+ \rho \cdot 48\mathfrak{D})\lceil \sqrt{s}\rceil + (2b + \rho \cdot 48 \mathfrak{D}-1)\lceil \sqrt{s}\rceil + 1$ colors,
    which implies $\Gamma \leq (4b+ 2 \rho \cdot 48\mathfrak{D} -1)\lceil \sqrt{s} \rceil + 1$. Since $s\geq 1$ and 
    $\rho \cdot 48\mathfrak{D} \leq \frac{3\mathfrak{D}}{\lceil \sqrt{s}\rceil}$, the 
    number of time units needed within the interval
    $I_{z}$ is at most:

  \begin{eqnarray*}
    &\mathfrak{D}&+ \ 2\mathfrak{D}+1 + \Gamma + \mathfrak{D} +1
    = 4\mathfrak{D} +2+  \\
    & &
    \left(4b  + 2 \frac{3\mathfrak{D}}{\lceil \sqrt{s}\rceil} - 1\right)\lceil \sqrt{s}\rceil +1
    \leq  4b \lceil \sqrt{s}\rceil + 12\mathfrak{D} \leq \frac{\tau}{4} \ . 
 \end{eqnarray*}

\end{itemize}
\end{proof}

\begin{lemma}
\label{lemma:centralized-scheduler-transaction-generation}
In Algorithm \ref{alg:centralized-optimized-scheduler},
for transaction generation rate
$\rho \leq \rho'$
and $b\geq 1$,
in any interval $I_z$ of length $\tau$, where $z \geq 1$, it holds that:\\
(i) transactions that are scheduled at the last event at $I_z$ will be committed by at most the middle of $I_{z+1}$, and 
(ii) all the transactions generated in the interval $I_z$ will
be committed or aborted by the end of $I_{z+1}$.
\end{lemma}
\begin{proof} 
    We prove property (i) and (ii) together by induction on the number $z$ of the interval.
    The induction basis trivially holds for $z=0$,
    as no transactions are generated at $I_0$.
    By the inductive hypothesis, assume that
    properties (i) and (ii) hold for any interval $I_j$, where $j\leq z-1$. A visual representation of the scheduling event is shown in Figure~\ref{fig:example_alg1}.

    \begin{figure}[!ht]
  \centering
  \includegraphics[width=0.49\textwidth]{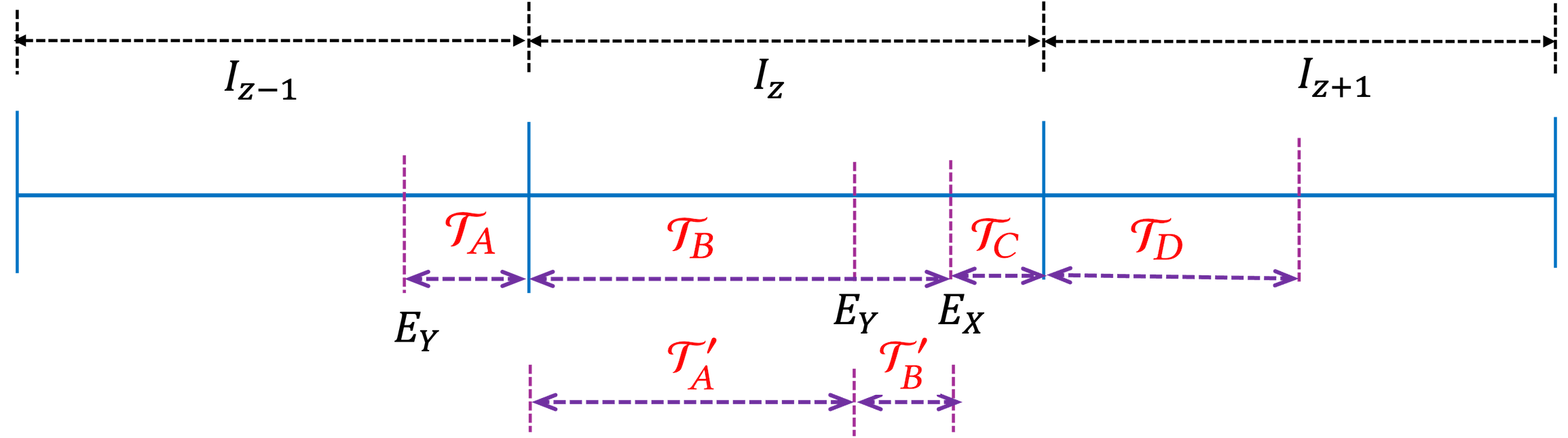}

  \caption{Simple representation of scheduling event triggered in Algorithm~\ref{alg:centralized-optimized-scheduler}.}
  \label{fig:example_alg1}
\end{figure}


\noindent
    \\{\bf\em Inductive step:}
    we prove the same property for interval $I_{z}$.

\begin{itemize}[leftmargin=*]
    \item  {\bf Case 1:} {\em No scheduling event occurs at $I_z$}.\\
        In this case, we show that no transaction can be generated at $I_z$.
        Suppose, for the sake of contradiction, that at least one transaction 
        has been generated at $I_z$. 
        Then, there are two subcases, and we obtain a contradiction in both.
        
        \begin{itemize}[leftmargin=*]
            \item {\bf Case 1.a:} {\em There was a scheduling event at $I_{z-1}$.}\\            
        Let $E_Y$ be the event at $I_{z-1}$. Based on the induction hypothesis, all the transactions scheduled at $E_Y$ will be committed by the middle of $I_z$.
        Since we assumed at least one transaction is generated at $I_z$, an event must be triggered in $I_z$ at the latest right after the middle of $I_z$. This contradicts that there is no event in $I_z$.

        \item {\bf Case 1.b:} {\em There was no scheduling event at $I_{z-1}$ (any scheduling event occurred at $I_{z-2}$ or earlier).}\\
        Case 1.b is not possible because of the following explanation. Let $E_L$ be the last event before $I_{z-1}$. By induction hypothesis, all transactions scheduled due to $E_L$ have already been committed by the middle of $I_{z-1}$ or earlier. If any transaction was generated after $E_L$, then that would immediately trigger an event at $I_z$ or earlier, which contradicts that there is no event after $E_L$. If $E_L$ does not exist, some event must be triggered at or before $I_{z}$, which is also a contradiction.

        \end{itemize}
        \item {\bf Case 2:} {\em A scheduling event occurs at $I_z$.}\\
        Let $E_X$ be the last event at $I_z$ (see Figure~\ref{fig:example_alg1}). There are two subcases:
        \begin{itemize}[leftmargin=*]
            \item  {\bf Case 2.a:} {\em The previous scheduling event $E_Y$ occurs~at interval~$I_{z-1}$}. 
            Let $\T_A \leq \zeta = (2b+\rho\cdot 48\mathfrak{D})s$ (from Lemma~\ref{lemma:max-injection-rate})
            be the number of transactions generated during the interval $I_{z-1}$ after event $E_Y$.
            From Lemma~\ref{lemma:time-required-for-scheduler-when-event-triggered}, the time required to commit these $\T_A$ transactions is at most $\frac{\tau}{4}$.
            By the induction hypothesis, all the $\T_A$ transactions will commit or abort by the end of $I_z$.

Let $\T_B \leq \zeta$ (from Lemma~\ref{lemma:max-injection-rate}) denote the number transactions generated in $I_z$ up to $E_X$. From 
 Lemma~\ref{lemma:time-required-for-scheduler-when-event-triggered}, the total time required to schedule the $\mathcal{T}_B\leq \zeta$ transactions is at most $\frac{\tau}{4}$. 
In event $E_X$ the transactions in $T_A$ and $T_B$ are being processed.
Since the transactions of $T_A$ are scheduled to commit or abort by the end of $I_z$,
the transactions of $T_B$ can be scheduled in the first $\frac{\tau}{4}$ time units of $I_{z+1}$. 
This proves the property (i) as $\frac{\tau}{4}<\frac{\tau}{2}$ (i.e. the middle of $I_{z+1}$).

\item 
\ra{\bf Case 2.b:} {\em The previous scheduling event $E_Y$ occurs at interval $I_z$.} From the above proof and our hypothesis, if there are stale transactions from $I_{z-1}$, they are committed within $I_z$. Therefore, the required schedule length in $I_{z+1}$ is only for the transactions that are generated in $I_z$. 

When the previous scheduling event $E_Y$ occurs at interval $I_z$ there are at most $\T_A' \leq \zeta$ transactions (from Lemma~\ref{lemma:max-injection-rate} and design of our algorithm), which corresponds to at most $\frac{\tau}{4}$ required time units (from Lemma~\ref{lemma:time-required-for-scheduler-when-event-triggered}) to get processed.
The transactions generated between events $E_Y$ and $E_X$ in $I_z$, 
are $\T_B' \leq \zeta$ (from Lemma~\ref{lemma:max-injection-rate}) 
and require $\frac{\tau}{4}$ time to be completed from Lemma~\ref{lemma:time-required-for-scheduler-when-event-triggered}).
Thus, the total scheduling time needed for all transactions generated in $I_z$ up to event $E_X$ is at most:
$
\frac{\tau}{4} + \frac{\tau}{4} = \frac{\tau}{2}.
$
Hence, all these transactions will be completed by the middle of $I_{z+1}$.
This proved property (i).

\item {\bf Case 2.c:} {\em The previous scheduling event $E_Y$ happens before $I_{z-1}$}.
In this case, all stale transactions from intervals $I_{z-2}$ or earlier
are being completed by the end of $I_{z-1}$ (by the inductive hypothesis).
Note that there could not be new transactions generated in $I_{z-1}$,
since this would trigger a scheduling event before the end of $I_{z-1}$.
Therefore, at the beginning of interval $I_z$, there are no stale transactions pending from previous intervals.

Let us now consider the transactions generated within $I_z$. Let $\T_B \leq \zeta$ (Lemma~\ref{lemma:max-injection-rate}) be the transactions generated in $I_z$ up to $E_X$. Similar to the previous case 2.b, the transactions scheduled by $E_X$ in $I_z$ are committed by the middle of $I_{z+1}$. This proved property (i).\\

        \end{itemize}
We continue to prove the property (ii) for Case 2. Let $\T_C$ be the number of transactions generated after event $E_X$ and before the end of interval $I_z$. 
From Lemma~\ref{lemma:max-injection-rate}, we have $\T_C \leq \zeta$, and by Lemma~\ref{lemma:time-required-for-scheduler-when-event-triggered}, processing of $\T_C$ requires at most $\frac{\tau}{4}$ time units. From the property (i) proved earlier, all of the transactions that are scheduled at the last event $E_X$ in $I_z$ are committed by at most the middle of the $I_{z+1}$. Thus, the scheduling event for the transactions $\T_C$ occurs by the middle of $I_{z+1}$, and requires at most $\frac{\tau}{4}$ time units to commit $\T_C$.
Moreover, when the scheduling event occurs in the middle of $I_{z+1}$, there may be new transactions generated in $I_{z+1}$ that are additional to $\T_C$; we denote this additional number of transactions as $\T_D$, where $\T_D \leq \zeta$ (from Lemma~\ref{lemma:max-injection-rate}). These transactions $\T_D$ also require at most $\frac{\tau}{4}$ time units to be committed (again by Lemma~\ref{lemma:time-required-for-scheduler-when-event-triggered}).
Therefore, the total schedule length required in $I_{z+1}$ is:
 $\frac{\tau}{2}$ for transactions scheduled at last event $E_X$,
$\frac{\tau}{4}$ for transactions $\T_C$, and
 $\frac{\tau}{4}$ for transactions $\T_D$.
Thus, the total time unit required is at most:
$
\frac{\tau}{2} + \frac{\tau}{4} + \frac{\tau}{4} = \tau.
$

 As the interval length $I_{z+1}$ is $\tau$, it is sufficient to process all the transactions generated in the interval $I_z$. This completes the inductive proof of property~(ii).
\end{itemize}%
\end{proof}

\begin{theorem}[Stability of {\sc EdSlScheduler}]
\label{theorem:single-leader}
    In Algorithm \ref{alg:centralized-optimized-scheduler}, for transaction generation rate 
 $\rho \leq \max\left\{ \frac{1}{16k}, \frac{1}{16\lceil \sqrt{s} \rceil} \right\}$ and burstiness $b \geq 1$, the number of combined pending transactions throughout the system at any given time is at most $2(2b+\rho\cdot 48\mathfrak{D})s$, and the transaction latency is at most $32b\cdot \min\{k, \sqrt{s}\} + 96\mathfrak{D}$.
\end{theorem}

\vspace{-2mm}
\begin{proof}
    To estimate the number of pending transactions during any time unit, consider a time unit within an interval $I_z$. 
    From the proof of Lemma \ref{lemma:max-injection-rate}, the maximum number of stale transactions during any time of $I_z$ is $\zeta \leq (2b+\rho\cdot 48\mathfrak{D})s$. Within $I_z$, there can be at most $(2b+\rho\cdot 48\mathfrak{D})s$ newly generated transactions. Therefore, the upper bound on pending transactions during any time unit is $2(2b+\rho\cdot 48\mathfrak{D})s$.
    For estimating transaction latency, we rely on the fact that a transaction generated in an interval will be processed by the end of the next interval, see Lemma \ref{lemma:centralized-scheduler-transaction-generation} property (ii). 
    Consequently, the transaction latency is bounded by twice the duration of the maximum interval length $\tau$:

    \[
    2 \tau = 2 ( 16b\cdot \min \{k, \lceil \sqrt{s}\rceil\} + 48 \mathfrak{D}) = 32b\cdot \min \{k, \lceil \sqrt{s}\rceil\}+96\mathfrak{D} .
    \]%

\end{proof}

\begin{theorem}
   \label{theorem:alg1-safety-and-liveness}
       Our proposed scheduling Algorithm~\ref{alg:centralized-optimized-scheduler} provides safety and liveness.
   \end{theorem}

    \begin{proof}
    We achieved safety by using a greedy vertex coloring algorithm within the leader shard, where the leader shard assigns a distinct color to each conflicting transaction. 
    In our protocol, each shard sends its transactions to the leader shard, which constructs a transaction conflict graph. 
    Then the leader assigns different colors to conflicting transactions and the same color to non-conflicting transactions using the greedy vertex coloring algorithm. After that, transactions are processed according to their assigned colors, which ensures that conflicting transactions commit at different time units.

    Liveness is guaranteed by ensuring that every generated transaction eventually commits. 
    There is no deadlock in Algorithm~\ref{alg:centralized-optimized-scheduler}, as only a single leader controls the shard without interference from other leaders. Moreover, from Theorems~\ref{theorem:single-leader}, we have a bounded transaction latency for the single-leader scheduler.
    This ensures that every generated transaction is either committed or aborted within a bounded number of time units.
   \end{proof}

\section{Event-Driven Scheduler with Multiple Leaders}
\label{sec:multiple-leaders}
A central authority, as considered in Section~\ref{sec:single-leader}, might not exist in some blockchain sharding since, in principle, each transaction could be generated independently in a distributed manner in any shard, and gathering knowledge about them could be time- and resource-consuming. Here, we discuss a fully distributed scheduling approach using a clustering technique that allows the transaction schedule to be computed in a decentralized manner without requiring a single central authority. Each cluster has its own leader, and within each cluster, the distributed approach invokes 
the Event-Driven Single-Leader Scheduling
Algorithm~\ref{alg:centralized-optimized-scheduler} from Section~\ref{sec:single-leader}.%

\subsection{Shard Clustering}
\label{shard-clustering}
Similarly to~\cite{adhikari2024spaastable}, we consider a multi-layer clustering~\cite{gupta2006oblivious} of shard graph $G_s$, which is calculated before the algorithm starts and known to all the shards. This clustering scheme has been previously used to optimize transaction scheduling and execution
~\cite{sharma2014distributed,busch2022dynamic,adhikari2024spaastable,adhikari2024fast}.
A key advantage of this approach is locality-awareness, where transactions accessing nearby shards are prioritized over those that involve more distant shards~\cite{adhikari2024spaastable,adhikari2024fast}.
The cluster is formed by dividing the shard graph $G_s$ into $H_1 = \lceil \log D \rceil +1$ layers of clusters (logarithms are base 2, and $D$ is the diameter of $G_s$). Each layer consists of multiple clusters, and each cluster contains a set of shards.
Layer $i$, where $0\leq i < H_1$, is a sparse cover of $G_s$, which has the following properties: 
(i) Each cluster at layer 
$i$ has a diameter of at most $O(2^i\log s)$;
(ii) Each shard belongs to at most
 $O(\log s)$ different clusters within layer $i$; 
(iii) For $i \geq 1$, for every shard $S_k$, at least one cluster at layer 
$i$ contains the entire $2^{i-1}$-neighborhood of $S_k$.

The layer 0 is a special case that has each shard as a cluster on its own.
Each layer $i \geq 1$ is further subdivided into $H_2 = O(\log s)$ partitions, following the sparse cover construction in~\cite{gupta2006oblivious}. 
These partitions, referred to as {\em sub-layers}, are labeled from $0$ to $H_2-1$.
A shard may be present in all 
$H_2$ sub-layers, ensuring that at least one of these clusters contains its complete $2^{i-1}$-neighborhood.

Within each cluster at layer $i \geq 1$, a leader shard is designated such that its $2^{i-1}$-neighborhood is entirely contained within the cluster. To systematically identify clusters across layers and sub-layers, we introduce the concept of height, represented as a tuple $(i, j)$, where $i$ refers to the layer and $j$ refers to the sublayer. Following~\cite{sharma2014distributed,busch2022dynamic,adhikari2024spaastable}, cluster heights are lexicographically ordered to maintain a structured hierarchy.
For a cluster $C$ let $height(C) = (i, j)$ (where $0 \leq i < H_1$ represents the layer, and $0 \leq j < H_2$ represents the sublayer),
denote its height.

The {\em home cluster} of a transaction $T_i$ is determined based on its home shard 
$S_{i}$ and the destination shards it needs to access.
Let $z$ be the maximum distance between $S_i$ and the destination shards that transaction $T_i$ accesses.
The home cluster of $S_i$ is defined as the lowest-layer and sub-layer (with lowest lexicographic order) cluster that contains the $z$-neighborhood of $S_i$.
Each home cluster has a designated leader shard responsible for handling all transactions originating from shards within that cluster. When a transaction 
$T_i$ is generated at its home shard, the transaction is sent to the leader shard of the corresponding home cluster, where the scheduling process is determined.
\subsection{Multi-leader Scheduling Algorithm} 
We provide the pseudocode of the Multi-leader scheduler in Algorithm~\ref{alg:multiple-leaders-scheduler}.
Consider a cluster $C$ with $height(C) = (i, j)$.
Within cluster $C$ we use Algorithm~\ref{alg:centralized-optimized-scheduler} to schedule and commit transactions. The leader $S_{ldr}$ of $C$ first appends any transactions it receives into its pending queue $PQ_{ldr}$. 
If the schedule length $\lambda$ of  $PQ_{ldr}$ transactions is at least the previously computed schedule length, denoted as $LastEventSchLength$, then the leader schedules and processes those transactions. To schedule transactions,
$S_{ldr}$ first acquires $scheduleControl$. Once it has $scheduleControl$ it moves transactions from $PQ_{ldr}$ to $SQ_{ldr}$ and processes those transactions similar to Algorithm~\ref{alg:centralized-optimized-scheduler}.

The acquisition of $scheduleControl$ is coordinated with 
``parent'' and ``child'' clusters.
For a cluster $C$ with $height(C) = (i,j)$ let $(i',j')$ and $(i'',j'')$ denote its immediate higher and lower heights,
namely, $(i',j') > (i,j) > (i'',j'')$
(assume that $(i,j)$ is not the top-most or bottom-most height,
as those can be treated as simpler special cases). 
A {\em parent} cluster of $C$ is any cluster $C'$ at the immediate above height $(i',j')$ that has common shards with $C$.
A {\em child} cluster of $C$ is any cluster $C''$ at the immediate below height $(i'',j'')$ that has common shards with~$C$.
Note that $C$ may have multiple parents or children.

The communication between $C$ and its parent/child clusters occurs through event-driven requests and responses between their leaders.
If $C$ is at the bottom-most height, $(i,j) = (0,0)$,
then $C$ has by default the scheduling control, namely, $scheduleControl(C) = true$.
However, if $C$ is not at the bottom-most height
then it initially does not have the scheduling control 
and it must acquire it from its children.
The \textit{schedule\_control\_request} event is used to request scheduling control from $C$ to all children clusters.
Once all the children have returned a positive \textit{schedule\_control\_response},
then $scheduleControl(C) = true$.
When this happens, the leader $S_{ldr}$ of $C$
proceeds with executing locally Algorithm~\ref{alg:centralized-optimized-scheduler}.
After that, $C$ relinquishes its control back to its children
by sending \textit{schedule\_control\_release}
which pass it down to their own children until it reaches the bottom-most height.

Each time when the leader shard $S_{ldr}$ of $C$ has pending transactions $PQ_{ldr}$ the schedule length $\lambda$ reaches the minimum $LastEventSchLength$,
then it requests the scheduling control from its children.
However, $C$ may request to obtain control on behalf of a parent.
If any of its parent $C'$ requests the scheduling control from $C$,
then if $C$ currently has the scheduling control, it responds positively to $C'$.
Otherwise, $C$ forwards the control acquisition request down to its children.
When all the children respond positively, 
then it passes the control to the parent $C'$.
While trying to obtain control for its parent,
$C$ may use the control to process its own transactions if
new transactions have arrived in leader $S_{ldr}$,
before passing it to the parent $C'$.

If cluster $C$ requests control acquisition for processing its own new transactions,
then it only needs to involve the children
which intersect with the destination shards of 
the transactions in the queue $PQ_{ldr}$ of its leader $S_{ldr}$.
This assures that there will be no interference between leaders that have no common intersecting destination shards in the transactions of their pending queues.
When a control acquisition request is forwarded by $C$ to children, we include in the request the involved destination shards,
in order to pass the request only to the children clusters that contain such destination shards.

\begin{algorithm*}
\smaller[1]
\caption{\sc EdmlScheduler}
\label{alg:multiple-leaders-scheduler}

Assume a hierarchical cluster decomposition shard graph $G_s$ is precomputed before the algorithm starts and known to all the shards\;
Each cluster $C$ belongs at some height $(i,j)$, at layer $i$ and 
sublayer $j$ of the hierarchical decomposition, where $0 \leq i < H_1$ and $0 \leq j < H_2$; the heights $(i,j)$ are ordered lexicographically\;
\tcp{notation for schedule\_control\_request and schedule\_control\_response}
schedule\_control\_request$(C(i',j'), C(i,j), accessingShards):$ $scheduleControl$ for $accessingShards$ request from cluster $C(i',j')$ to cluster $C(i,j)$\;
schedule\_control\_response$(C(i,j), C(i',j'), accessingShards):$ $scheduleControl$ for $accessingShards$ response from cluster $C(i,j)$ to cluster $C(i',j')$\;


\BlankLine
\BlankLine
\tcc{Event on Home Shard}
\SetKwBlock{EventOne}{\bf Upon event: {\em new transaction $T_i$ generated at home shard $S_i$}}{}
\EventOne{
    Identify home cluster $C(i,j)$ containing $S_i$ and destination shards accessed by $T_i$\;
    Shard $S_i$ {\em \textbf{Emit} transaction\_send}($T_i$) to cluster leader shard $S_{ldr}$ of  $C(i,j)$\;
}

\BlankLine
\BlankLine
\tcc{Event: Leader Shard Receives Transaction or Scheduling Control}

\SetKwBlock{EventTwo}{\bf Upon event: {\em transaction\_send}($T_i$) received at leader shard $S_{ldr}$ at cluster $C(i,j)$ OR {\em schedule\_control\_response($C(i',j'), C(i,j), accessingShards$)} OR {\em schedule\_control\_response($C(i'',j''), C(i,j), accessingShards$)}}{}

\EventTwo{
    Append $T_i$ (if any) to the pending queue $PQ_{ldr}$ of cluster leader shard at $C(i,j)$\;
     Constructs transaction conflict graph and colors it using greedy vertex coloring to determine schedule length $\lambda$ for transactions in $PQ_{ldr}$\;
    
    \If{$\lambda \geq LastEventSchLength$}{
        \If{$scheduleControl$ is at cluster $(i,j)$}{
        Move transactions  from $PQ_{ldr}$ to schedule queue of leader shard $SQ_{ldr}$ at cluster $C(i,j)$\;
            Process all transactions $\T'$ of $SQ_{ldr}$ using Algorithm~\ref{alg:centralized-optimized-scheduler}\;
           Wait until all transactions $\T'$ are committed\;
            Release $scheduleControl$ to shards access by $\T'$\;
            
                \tcc{Send $scheduleControl$ back to child clusters}
            \ForEach{child cluster $C(i'', j'')$ of $C(i,j)$}{
                    {\em \textbf{Emit} schedule\_control\_response}($C(i,j), C(i'',j''), accessingShards$)\;
                }
        }
        \Else{
            \tcc{Request control from all child clusters simultaneously}
            \ForEach{child cluster $(i'', j'')$ of $C(i,j)$}{
                {\em \textbf{Emit} schedule\_control\_request}($C(i,j), C(i'',j''), accessingShards$)\;
            }
        }
    }
    \Else{
          \tcc{Send $scheduleControl$ back to child clusters}
            \ForEach{child cluster $(i'', j'')$ of $C(i,j)$}{
                {\em \textbf{Emit} schedule\_control\_response}($C(i,j), C(i'',j''), accessingShards$)\;
            }
    }
}
\BlankLine
\BlankLine
\tcc{Leader shard at cluster $C(i,j)$ receive the message from immediate parent cluster $C(i',j')$ }
\SetKwBlock{EventThree}{\bf Upon event:  {\em  schedule\_control\_request}($C(i',j'), C(i,j), accessingShards$)}{}
\EventThree{
   \If{$scheduleControl$ is in $C(i,j)$ for all $accessingShards$}{
         \If{All transactions $\T'$ at leader of $C(i,j)$ are committed}{
         Release $scheduleControl$ to shards access by $\T'$\;
          Set $scheduleControl$ to each $accessingShards$\;
            {\em \textbf{Emit} schedule\_control\_response}($C(i,j), C(i',j'), accessingShards$)\;
        }
        \Else{
            Wait until all scheduled transactions  $\T'$ are committed\;
            Release $scheduleControl$ to shards access by $\T'$\;
            Set $scheduleControl$ to each $accessingShards$\;
            {\em \textbf{Emit} schedule\_control\_response}($C(i,j), C(i',j'), accessingShards$)\;
        }
   }
   \Else{
    \tcc{Request schedule control to all child clusters recursively}
        \ForEach{child cluster $(i'', j'')$ of $C(i,j)$}{
            {\em \textbf{Emit} schedule\_control\_request}($C(i,j), C(i'',j''), accessingShards$)\;
        }
    }
}

\end{algorithm*}

 \subsection{Analysis for Multi-leader Scheduler}
The multi-leader scheduler extends the single-leader scheduling algorithm (Algorithm~\ref{alg:centralized-optimized-scheduler}) while introducing an additional overhead cost due to its hierarchical clustering structure. This overhead is 
$O(\log D \log s)$, where 
$O(\log D)$ accounts for the layered hierarchy and 
$O(\log s)$ comes from the sublayers. Here, 
$D$ represents the diameter of shard graph $G_s$, and 
$s$ denotes the total number of shards. 
Moreover, each cluster $C$ at level $(i,j)$ has its own diameter $\mathfrak{D}_i \leq \mathfrak{D}$, where $\mathfrak{D}$ is the maximum diameter among all clusters, which represents the upper bound on local processing time and communication delay between any two shards. This upper bound $\mathfrak{D}$ is defined in Section~\ref{preliminaries}.


For the sake of analysis of the multi-leader scheduling algorithm, we divide time into a sequence of intervals $I_1, I_2, \ldots,$ each of duration $\tau'$ as defined below.
To simplify the analysis, we also have an artificial interval $I_0$ before time begins with no transactions. Let the maximum transaction injection rate be $\rho''$, as defined below, where $c_1$ is some combined positive constant that comes from layers and sublayers,
and let $\zeta'$ denote the maximum number of transactions that can be generated in an interval:
\begin{eqnarray*}
    \tau' & := & 16c_1b \log D \log s \cdot \min\{k, \sqrt{s}\} + 48 c_1 \mathfrak{D} \log D \log s \ ,\\ 
    \rho''& := & \frac{1}{16c_1 \log D \log s}\max\left\{ \frac{1}{k}, \frac{1}{\lceil \sqrt{s} \rceil} \right\} \ ,
    \zeta' :=  (2b+ \rho 48 c_1 \mathfrak{D} \log D \log s)s \ .
\end{eqnarray*}

\begin{lemma}
\label{lemma:max-injection-rate-for-multi-leader}
In Algorithm \ref{alg:multiple-leaders-scheduler},
for transaction generation rate
$\rho \leq \rho''$
and $b\geq 1$,
in any interval $I_z$ of length $\tau'$, where $z \geq 1$, there are at most $\zeta'$ new transactions generated within~$I_z$.
\end{lemma}
\begin{proof}
    According to the definition of the adversary,
    during any interval $I_z$ of length $\tau'$, the maximum congestion added to any shard is:
 \begin{equation*}
    \label{eqn:transaction-generation-rate-alg2}
    \rho \cdot \tau' + b \leq b + \rho \cdot 48 c_1 \mathfrak{D} \log D \log s +b = 2b+ \rho\cdot 48 c_1 \mathfrak{D} \log D \log s
    \ 
    \end{equation*}

\noindent
        Since each transaction accesses at least one shard and there are at most $s$ shards, 
        the total number of new transactions generated in  $I_z$~is~at~most
$     (2b+ \rho\cdot 48 c_1 \mathfrak{D} \log D \log s)s = \zeta'$.
\end{proof}

\begin{lemma}
\label{lemma:time-required-for-multi-leader-scheduler-when-event-triggered}
In Algorithm~\ref{alg:multiple-leaders-scheduler},
for transaction generation rate $\rho \leq \rho''$
and $b\geq 1$
in any interval $I_z$ of length $\tau'$ (where $z \geq 1$), if there is only one cluster $C$ and there are at most $\zeta'$ pending transactions, then it takes at most $\frac{\tau'}{4c_1 \log D \log s}$ time units to commit all $\zeta'$ transactions.
\end{lemma}

\begin{proof}
  We calculate the time required to process $\zeta'$ transactions 
             based on the magnitude~of~$k$~and~$\lceil\sqrt{s}\rceil$.
             \begin{itemize}[leftmargin=*]
                 \item {\bf Case 1:  $k \leq \lceil \sqrt{s} \rceil$.}\\
    From Equation \ref{eqn:transaction-generation-rate-alg2} in the proof of Lemma~\ref{lemma:max-injection-rate-for-multi-leader}, the congestion of transactions per shard is $2b+\rho\cdot 48 c_1 \mathfrak{D} \log D \log s$.  
    Since each transaction accesses at most $k$ shards, each 
    pending
    transaction conflicts with at most $((2b+\rho\cdot 48 c_1 \mathfrak{D} \log D \log s)-1)k$ other
    pending
    transactions. Consequently, the highest degree $\Delta$ in the transactions conflict graph $G_\T$ 
    is $\Delta \leq ((2b+\rho\cdot 48 c_1 \mathfrak{D} \log D \log s)-1)k$. Hence, a greedy vertex coloring algorithm on $G_\T$ assigns at most $\Delta+1 \leq ((2b+\rho\cdot 48 c_1 \mathfrak{D} \log D \log s) -1)k+1$ colors. Since $k \geq 1$, the total length of the required time interval to schedule and commit/abort the transactions is calculated as follows. Assuming schedule control is currently in the cluster $C$ and the algorithm consists of different steps, the initial step is to send the information about the transaction from the home shard to the leader shard of cluster $C$, which takes at most $\mathfrak{D}$ time units.
    Similarly, the leader shard takes $1$ time unit for graph construction and coloring (as we consider local consensus time $\delta_{cons}=1$).
    Moreover, the leader asks the destination shards for the states of their accounts, 
    and the destination shards respond with the states of their accounts. This combined takes $2\mathfrak{D}$ time units. Additionally, leader shards take $\Delta+1$ time units to pre-commit each conflicting transaction serially, and $\mathfrak{D}$ time units to send the committed subtransaction batch to the destination shard. Finally, the destination shard takes $1$ time unit to reach a consensus on the received order of subtransactions batch and appends on their local blockchain. Since $\rho \leq \rho''$, we have $\rho\cdot 48 c_1 \mathfrak{D} \log D \log s \leq \frac{3\mathfrak{D}}{k}$.
    Hence, the total combined required time length can be calculated~as

   \begin{eqnarray*}
         \mathfrak{D}\ + &2\mathfrak{D}&+\ 1 + (\Delta+1) + \mathfrak{D} +1 
         \\
         &\leq& 4\mathfrak{D} + 2 + 2bk +\frac{3\mathfrak{D}}{k}k -k +1 \\ 
        &\leq&  2bk + 9\mathfrak{D} 
        \\
        &\leq& \frac{\tau'}{4c_1 \log D \log s}\ .
    \end{eqnarray*}

\item   {\bf Case 2: $k> \lceil \sqrt{s}\rceil$.}\\
     We classify the $\zeta'$ 
     pending transactions into two groups, the ``heavy'' transactions, which access more than $\lceil \sqrt{s}\rceil$ shards, and ``light'' transactions, which access at most $\lceil \sqrt{s}\rceil$ shards.
       The maximum number of heavy 
       transactions can be $(2b+\rho\cdot 48 c_1 \mathfrak{D} \log D \log s)\lceil \sqrt{s}\rceil$. If there were more, the total congestion 
       would be strictly greater than $(2b+\rho\cdot 48 c_1 \mathfrak{D} \log D \log s)\lceil \sqrt{s}\rceil \cdot \sqrt{s} \geq (2b+\rho\cdot 48 c_1 \mathfrak{D} \log D \log s)s$, which is not possible 
       \ra{because of Lemma~\ref{lemma:max-injection-rate-for-multi-leader}}.
    A coloring algorithm for the conflict graph $G_\T$ can assign each of the heavy transactions a unique color, which can be done using at most $\Gamma_1 =(2b+\rho\cdot 48 c_1 \mathfrak{D} \log D \log s)\lceil \sqrt{s}\rceil$~colors.\\
        
    The remaining transactions are light. Let $G'_\T$ be the subgraph of the conflict graph $G_\T$ 
   induced by
    the light transactions. Each light transaction conflicts with at most $((2b+\rho\cdot 48 c_1 \mathfrak{D} \log D \log s)-1)\lceil \sqrt{s}\rceil$
     other light transactions.
     Hence, 
     the degree of $G'_\T$ is at most $(2b+\rho\cdot 48 c_1 \mathfrak{D} \log D \log s-1)\lceil \sqrt{s}\rceil$.
     Thus, $G'_\T$ can be colored with at most $\Gamma_2 = (2b+\rho\cdot 48 c_1 \mathfrak{D} \log D \log s-1)\lceil \sqrt{s}\rceil + 1$~colors.

    Consequently, the conflict graph $G_\T$ can be colored with at most $\Gamma = \Gamma_1+\Gamma_2 = (2b+\rho\cdot 48 c_1 \mathfrak{D} \log D \log s)\lceil \sqrt{s}\rceil + (2b+\rho\cdot 48 c_1 \mathfrak{D} \log D \log s-1)\lceil \sqrt{s}\rceil + 1$ colors,
    which implies $\Gamma \leq (4b+ 2 \rho\cdot 48 c_1 \mathfrak{D} \log D \log s -1)\lceil \sqrt{s} \rceil + 1$. Since $s\geq 1$ and 
    $\rho\cdot 48 c_1 \mathfrak{D} \log D \log s \leq \frac{3\mathfrak{D}}{\lceil \sqrt{s}\rceil}$, the 
 number of time units needed within the interval
    $I_{z}$ is at most:

     \begin{eqnarray*}
          \mathfrak{D}\ + &2\mathfrak{D}&+\ 1 + \Gamma + \mathfrak{D} +1 
          \\
    &=& 4\mathfrak{D} +2+  \left(4b  + 2 \frac{3\mathfrak{D}}{\lceil \sqrt{s}\rceil} - 1\right)\lceil \sqrt{s}\rceil +1\\   
        &\leq&  4b \lceil \sqrt{s}\rceil + 12\mathfrak{D}  
        \\
        &\leq& \frac{\tau'}{4c_1 \log D \log s} \ .
    \end{eqnarray*}

\end{itemize}
\end{proof}

\begin{lemma}
\label{lemma:time-required-for-multi-leader-scheduler-when-event-triggered-for-all-cluster}
In Algorithm~\ref{alg:multiple-leaders-scheduler},
for transaction generation rate $\rho \leq \rho''$
and $b\geq 1$
in any interval $I_z$ of length $\tau'$ (where $z \geq 1$), if there are at most $\zeta'$ pending transactions considering all layers and sublayers, then it takes at most $\frac{\tau'}{4}$ time units to commit all $\zeta'$ transactions.
\end{lemma}
\begin{proof}
According to the algorithm design, each cluster leader requires schedule control to schedule their transactions. Suppose cluster $C$ is not at the bottom-most height
then it initially does not have the scheduling control 
and it must be acquired schedule control from its children.
The \textit{schedule\_control\_request} event is used to request scheduling control from $C$ to all children clusters. However,  children clusters could also have transactions and reach a threshold to schedule the transactions and access the same destination shards with $C$; thus, they may also request the schedule control and append transactions for scheduling.

As the multi-layer clustering~\cite{gupta2006oblivious} of shards consists of $O(\log D)$ layers and $O(\log s)$ sublayers, each child cluster layer and sublayer may have the transaction to schedule and access the same destination shard $S_j$. So the schedule of cluster $C$ is restricted by the combined schedule of all the sublayers below.
From Lemma \ref{lemma:time-required-for-multi-leader-scheduler-when-event-triggered}, cluster $C$ and also every descendant cluster of $C$ induces scheduling length at most $\frac{\tau'}{4c_1 \log D \log s}$.
Since there are $c_1 \log D \log s$ sublayers, the total scheduling length from $C$ and the descendant clusters from all sublayers below it is at most $\frac{\tau'}{4c_1\log D \log s} \cdot c_1 \log D \log s \leq \frac{\tau'}{4}$. 
\end{proof}

We now restate Lemma~\ref{lemma:centralized-scheduler-transaction-generation} for Algorithm~\ref{alg:multiple-leaders-scheduler}. 
The proof follows the same structure as for Algorithm~\ref{alg:centralized-optimized-scheduler}, with one change: 
Every time when the scheduling event occurs, instead of using Lemma~\ref{lemma:time-required-for-scheduler-when-event-triggered}, 
we use Lemma~\ref{lemma:time-required-for-multi-leader-scheduler-when-event-triggered-for-all-cluster}. 
With this change, the rest of the proof remains the same, and we get the following Lemma~\ref{lemma:multi-leader-scheduler-txn-processing}.

\begin{lemma}
\label{lemma:multi-leader-scheduler-txn-processing}
In Algorithm \ref{alg:multiple-leaders-scheduler},
for transaction generation rate
$\rho \leq \rho''$
and $b\geq 1$,
in any interval $I_z$ of length $\tau'$, where $z \geq 1$, it holds that:
(i) transactions that are scheduled at the last event at $I_z$ will be committed by at most the middle of $I_{z+1}$, and 
(ii) all the transactions generated in the interval $I_z$ will
be committed or aborted by the end of $I_{z+1}$.
\end{lemma}

\begin{theorem}[Multi-leader stability]
    \label{theorem:multi-leader-scheduler}
In the multi-leader scheduler, for injection rate 
$\rho \leq \frac{1}{16 c_1 \log D \log s} \cdot \max\left\{ \frac{1}{k}, \frac{1}{\sqrt{s}} \right\}$,
 with burstiness $b \geq 1$, the number of combined pending transactions throughout the system at any given time unit is at most $2(2b+ \rho \cdot 48 c_1 \mathfrak{D} \log D \log s)s$, and the transaction latency is upper bounded by at most  $32c_1b \log D \log s \cdot \min\{k, \sqrt{s}\} + 96 c_1 \mathfrak{D} \log D \log s$.
\end{theorem}

\begin{proof}
    To estimate the number of pending transactions during any time unit, consider a time unit within an interval $I_z$. 
    From the proof of Lemma~\ref{lemma:max-injection-rate-for-multi-leader}, the maximum number of stale transactions during any time of $I_z$ is $(2b+ \rho \cdot 48 c_1 \mathfrak{D} \log D \log s)s$. Within $I_z$, there can be at most $(2b+ \rho \cdot 48 c_1 \mathfrak{D} \log D \log s)s$ newly generated transactions. Therefore, the upper bound on pending transactions during any time unit is $2(2b+ \rho \cdot 48 c_1 \mathfrak{D} \log D \log s)s$.
    For estimating transaction latency, we rely on the fact that a transaction generated in an interval will be processed by the end of the next interval, see Lemma \ref{lemma:multi-leader-scheduler-txn-processing}. 
    Consequently, the transaction latency is bounded by twice the duration of the maximum interval length $\tau'$:  
    \[
    2 \tau' = 32c_1b \log D \log s \cdot \min\{k, \sqrt{s}\} + 96 c_1 \mathfrak{D} \log D \log s.
    \]%
    
%
\end{proof}
\begin{theorem}
\label{theorem:no-deadlock}
   There is no deadlock in the Multi-Leader Scheduling Algorithm~\ref{alg:multiple-leaders-scheduler}.
\end{theorem}

\begin{proof}
For the sake of contradiction, suppose that there exists a deadlock in Algorithm~\ref{alg:multiple-leaders-scheduler}.
This means that a set of clusters are involved in a circular wait. More specifically, suppose there is 
 existence of a cycle of clusters $C_1, C_2, \ldots, C_r$
 such that each cluster $C_i$ 
  is waiting for $scheduleControl$ from the next cluster $C_{i+1}$, and similarly the last cluster $C_r$ waiting for s$cheduleControl$ from first cluster $C_1$, and this forms a circular dependency.

  In Algorithm~\ref{alg:multiple-leaders-scheduler}, clusters are ordered lexicographically according to their height, and the acquisition of $scheduleControl$ by any cluster $C$ is coordinated with its {\em parent} and {\em child} clusters based on their heights. If $height(C) = (i,j)$, then the immediate higher height is $(i',j') > (i,j)$ and the immediate lower height is $(i'',j'') < (i,j)$. This hierarchical structure ensures that there is a strict ordering in which $scheduleControl$ can be acquired and released.
  
  According to the design of Algorithm~\ref{alg:multiple-leaders-scheduler}, a cluster $C$ may only wait for $scheduleControl$ from its {\em child} clusters at the immediate lower height. Thus, if $C_i$ is waiting for $C_{i+1}$, it must be that: $height(C_{i+1})< height(C_i)$ i.e. $C_{i+1}$ is at a strictly lower height than $C_i$. 
  
  From the cycle assumption we consider, we could have the sequence of the following inequalities:
  $height(C_2)< height(C_1)$, $height(C_3)< height(C_2)$, $\ldots$, $height(C_1)< height(C_r)$.

  By transitivity of the height ordering, this implies:
   $height(C_1)> height(C_r)> height(C_1)$,
   which is a contradiction, because a strict inequality cannot hold between an element and itself. Therefore, the assumption we made that a deadlock occurs leads to a logical contradiction. Thus, no deadlock can occur in Algorithm~\ref{alg:multiple-leaders-scheduler}. \end{proof}

   \begin{corollary}
   \label{corollary:safety-and-liveness}
       Our proposed scheduling Algorithm~\ref{alg:multiple-leaders-scheduler} provides safety and liveness.
   \end{corollary}

    \begin{proof}
    We achieved safety by using a greedy vertex coloring algorithm within the leader shard. 
    In our protocol, each shard sends its transactions to the leader shard, which constructs a transaction conflict graph. 
    Then the leader assigns different colors to conflicting transactions and the same color to non-conflicting transactions using the greedy vertex coloring algorithm. After that, transactions are processed according to their assigned colors, which ensures that conflicting transactions commit at different time units.

    Liveness is guaranteed by ensuring that every generated transaction eventually commits. 
    From Theorem~\ref{theorem:no-deadlock}, we have shown that there is no deadlock in Algorithm~\ref{alg:multiple-leaders-scheduler}. 
    Similarly, there is no deadlock in Algorithm~\ref{alg:centralized-optimized-scheduler}, as only a single leader controls the shard without interference from other leaders. 
    In fact, Algorithm~\ref{alg:centralized-optimized-scheduler} can be viewed as a special case of Algorithm~\ref{alg:multiple-leaders-scheduler} with only one cluster.

    Furthermore, from Theorems~\ref{theorem:single-leader} and~\ref{theorem:multi-leader-scheduler}, we bound the transaction latency for both single-leader and multi-leader schedulers.
    This ensures that every generated transaction is either committed or aborted within a bounded number of time units.
   \end{proof}

\section{Conclusions}
\label{sec:conclusion}
In conclusion, we designed a stable transaction scheduling algorithm for the blockchain sharding system. Our algorithm significantly improves the stable transaction injection rate compared to the best-known result in~\cite{adhikari2024spaastable}. Moreover, our proposed scheduler works in a partially-synchronous model, which is more practical than the synchronous model considered in~\cite{adhikari2024spaastable}.
To our knowledge, the bound we establish here provides the most efficient stable transaction rate for blockchain sharding systems. Blockchain designers can use our results to build systems that are resilient to DoS attacks.
   
   Note that our proposed schedulers eliminate the back-and-forth communication between the leader and destination shards used in~\cite{adhikari2024spaastable} for transaction confirmation and commitment. 
   This is achieved by pre-fetching account states and performing local pre-commitment at the leader shard. Since only valid transactions are pre-committed, no confirmation step from destination shards is required in our schedulers. 
   This change is one of the reasons we obtain near-optimal injection rate. 
In future work, we plan to simulate our proposed scheduling algorithm to validate our theoretical results. 
Further research directions could be dynamic shard clustering for a multi-leader scheduler, as our current model assumes that clusters are precomputed.

\begin{acks}
This paper is supported by NSF grant CNS-2131538.
\end{acks}

\bibliographystyle{ACM-Reference-Format}
\balance
\bibliography{references}

\end{document}